%% file: RM_treewidth_paper.tex
\documentclass[11pt]{article}
\usepackage{amsmath,amsfonts,amssymb,amsthm,color}

\usepackage{times}
\usepackage{epsfig}

\title{The Treewidth of MDS and Reed-Muller Codes\footnotemark[1]}
\author{Navin Kashyap\footnotemark[2] \and Andrew Thangaraj\footnotemark[3]}

\date{}

\newtheorem{theorem}{Theorem}
\newtheorem{lemma}[theorem]{Lemma}
\newtheorem{prop}[theorem]{Proposition}

\def\F{\mathbb{F}}
\def\Z{\mathbb{Z}}

\def\cC{{\mathcal C}}
\def\cD{{\mathcal D}}
\def\cM{{\mathcal M}}
\def\cT{{\mathcal T}}

\def\b{\mathbf{b}}

\def\w{\mathbf{w}}
\def\1{\mathbf{1}}

\def\hu{\hat{u}}
\def\vu{\check{u}}

\def\k{\kappa}
\def\om{\omega}
\def\s{\sigma}

\begin{document}
\maketitle

\renewcommand{\thefootnote}{\fnsymbol{footnote}}

\footnotetext[1]{This work was supported in part by a Discovery Grant
from the Natural Sciences and Engineering Research Council (NSERC), Canada.}
\footnotetext[2]{N.\ Kashyap is with the Department of Electrical Communication
Engineering, Indian Institute of Science, Bangalore, and with the 
Department of Mathematics \& Statistics, Queen's University, Kingston, Ontario,
Canada. Email: \texttt{nkashyap@ece.iisc.ernet.in}}
\footnotetext[3]{A.\ Thangaraj is with the Department of Electrical
Engineering, Indian Institute of Technology, Madras. 
Email: \texttt{andrew@ee.iitm.ac.in}}

\renewcommand{\thefootnote}{\arabic{footnote}}

\begin{abstract}
The constraint complexity of a graphical realization of a linear code
is the maximum dimension of the local constraint codes in the realization. 
The treewidth of a linear code is the least constraint complexity 
of any of its cycle-free graphical realizations. This notion provides
a useful parametrization of the maximum-likelihood decoding complexity 
for linear codes. In this paper, we prove the surprising fact that
for maximum distance separable codes and Reed-Muller codes, treewidth
equals trelliswidth, which, for a code, is defined to be the least 
constraint complexity (or branch complexity) of any of its 
trellis realizations. From this, we obtain exact expressions for the 
treewidth of these codes, which constitute the only known explicit expressions
for the treewidth of algebraic codes. 
\end{abstract}


\section{Introduction\label{sec:intro}}

A (normal) graphical realization of a linear code $\cC$ consists of an
assignment of the coordinates of $\cC$ to the vertices of a graph,
along with a specification of linear state spaces and linear
``local constraint'' codes to be associated with the edges and vertices,
respectively, of the graph \cite{For01}. Cycle-free graphical realizations, 
or simply \emph{tree realizations}, are those in which the underlying graph
is a tree. Tree realizations of linear codes are interesting because 
the sum-product algorithm (SPA) on such a realization is an 
exact implementation of maximum-likelihood (ML) decoding \cite{wiberg}.
The notion of constraint complexity of a tree realization 
was introduced by Forney \cite{For03} as a measure of the 
computational complexity of the corresponding SPA algorithm.
It is defined to be the maximum dimension among the local constraint codes 
constituting the realization. The \emph{treewidth} of a linear code
is the least constraint complexity of any of its tree realizations.

The minimal tree complexity measure defined for linear codes 
by Halford and Chugg \cite{HC08} is a close relative of treewidth.
There are also closely related notions of treewidth defined for graphs 
\cite{Bod93} and matroids \cite{HW06}; these relationships are discussed 
in more detail in \cite{Kas09a}. Known facts about the treewidth of graphs 
and matroids imply that computing the treewidth of a code is NP-hard.

For a length-$n$ linear code over the field $\F_q$,
the computational complexity of implementing ML decoding,
via the SPA on an optimal tree realization, is $O(n q^t)$, 
where $t$ is the treewidth of the code \cite{Kas09a}. 
In particular, ML decoding is fixed-parameter tractable with respect 
to treewidth, which means that for codes whose treewidth is bounded 
by a fixed constant $t$, ML decoding can be performed in polynomial time.
Thus, treewidth provides a useful parametrization of ML decoding complexity.

Trellis representations (or trellis realizations) 
of codes are special cases of tree realizations
which have received extensive attention in the literature
(see e.g., \cite{Var98}). In the context of trellis representations,
constraint complexity is usually called branch complexity. We define
here the \emph{trelliswidth} of a code to be the least branch complexity
of any of its trellis representations (optimized over all possible
orderings of the coordinates of the code). As trellis representations
are instances of tree realizations, trelliswidth is at least as large
as treewidth. In fact, it is known that trelliswidth 
can be much larger than treewidth: it was shown in \cite{Kas09b} that
the ratio of trelliswidth to treewidth can grow at most logarithmically 
with blocklength, and that there are codes with arbitrarily large blocklengths 
that achieve this logarithmic growth rate. The only known code family 
achieving logarithmic growth rate of this ratio is a family consisting 
of cut-set codes of a certain class of graphs. The codes in this family 
all have treewidth equal to 2, and rate approximately $1/4$, 
but minimum distance only 4 \cite{Kas09a}. 

It is not known if there are any other code families for which there is 
a significant advantage to be gained in going from trellis representations
to tree realizations that are topologically more complex. 
In the only previous investigation reported on this question, 
Forney \cite{For03} considered the family of Reed-Muller codes.
He showed that for a certain natural tree realization
of Reed-Muller codes, obtained from their well-known recursive $|u|u+v|$
construction, the constraint complexity is, in general, strictly larger 
than the trelliswidth of the code. But this still leaves open the possibility
that there may be other tree realizations whose constraint complexity 
beats trelliswidth. In particular, it leaves undecided the question
of whether the treewidth of a Reed-Muller code can be strictly less 
than its trelliswidth.

In this paper, we show that for Reed-Muller codes, treewidth 
is equal to trelliswidth. The proof of this makes use of 
structural properties known for optimal trellis realizations of 
Reed-Muller codes, and also relies strongly on a certain separator theorem
for trees. A similar proof strategy also works on the much simpler case of
maximum distance separable (MDS) codes, where again we show that
treewidth equals trelliswidth. These results yield the first explicit 
expressions for the treewidth of classical algebraic codes.

The rest of this paper is organized as follows. After providing the
necessary definitions and notation in Section~\ref{sec:prelims}, 
we describe, in Section~\ref{sec:strategy}, our proof strategy 
for showing that treewidth equals trelliswidth for certain codes.
Sections~\ref{sec:mds} and \ref{sec:RM} deal with MDS and Reed-Muller codes, 
respectively. The technical details of some of the proofs
are given in appendices.

\section{Preliminaries and Notation\label{sec:prelims}}

The notation $[n]$ denotes the set of positive integers from 1 to $n$;
$[a,b]$ denotes the set $\{i \in \Z: a \leq i \leq b\}$.
An $(n,k)$ linear code is a code of length $n$ and dimension $k$. 
The $n$ coordinates of the code are indexed by the elements of an index set 
$I$; unless specified otherwise, $I = [n]$. Given a linear code $\cC$ 
with index set $I$, for $J=\{j_1,j_2,\ldots,j_s\}\subseteq I$, 
the shortening of $\cC$ to the coordinates in $J$ is denoted $\cC_J$
and defined as follows:
$$\cC_J
  =\{c_{j_1} c_{j_2} \ldots c_{j_s}:\ c_1 c_2 \ldots c_n \in \cC, \, c_i=0\text{ for }i\notin J\}.$$ 

The notions of treewidth and trelliswidth are central to this article,
and we define these next.

\subsection{Treewidth and trelliswidth}

For brevity, we provide only the necessary definitions and main results;
for details, see \cite{For03},\cite{Kas09a}.

A tree is a connected graph with no cycles. The set of nodes and
the set of edges of a tree $T$ are denoted by $V(T)$ and $E(T)$, respectively.
Degree-1 nodes in a tree are called \emph{leaves},
and all other nodes are called \emph{internal nodes}.
We let $L(T)$ denote the set of leaves of $T$. 
A tree is a \emph{path} if all its internal nodes have degree 2; 
and is a \emph{cubic tree} if all its internal nodes have degree 3. 
A path with at least one edge has exactly two leaves; 
a cubic tree with $n$ leaves has $n-2$ internal nodes.

Let $\cC$ be an $(n,k)$ linear code with index set $I$. 
A \emph{tree decomposition} of $\cC$ is a pair $(T,\om)$, where $T$ is a tree
and $\om: I \to V(T)$ is an assignment of coordinates of $\cC$ to the nodes of 
$T$. 

Given a tree decomposition $(T,\om)$ of $\cC$, for each node $v$ of $T$,
we define a quantity $\k_v$ as follows. Let $E(v)$ denote the set of edges 
of $T$ incident on $v$. For $e \in E(v)$, let $T_{e,v}$ denote the component 
of $T-e$ ($T$ with $e$ removed) \underline{not} containing $v$. 
Finally, let $I_{e,v} = \om^{-1}(V(T_{e,v}))$ be the set of coordinates of 
$\cC$ that are assigned to nodes in $T_{e,v}$. Then,
\begin{equation}
  \label{eq:1}
  \k_v=k-\sum_{e\in E(v)}\dim(\cC_{I_{e,v}}).
\end{equation}
The quantity $\k_v$ above is the dimension of the local constraint code
at node $v$ in the minimal realization of $\cC$ on $(T,\om)$, denoted by
$\cM(\cC;T,\om)$.

Let $\k(\cC;T,\om) = \displaystyle \max_{v \in V(T)} \k_v$ denote 
the constraint complexity of $\cM(\cC;T,\om)$.
The treewidth of a code $\cC$, denoted by $\k(\cC)$, is then defined as 
\begin{equation}
\k(\cC) 
= \min_{(T,\omega)} \k(\cC;T,\om).
\label{treewidth_def}
\end{equation}
It is, in fact, enough to perform the minimization in (\ref{treewidth_def})
over cubic trees $T$ with $n$ leaves, and mappings $\om$ that 
are bijections between $I$ and $L(T)$.

The trelliswidth of $\cC$, which we will denote by $\tau(\cC)$, 
can be defined using the above notation as follows:
\begin{equation}
\tau(\cC) = \min_{\pi} \k(\cC;P,\pi),
\label{trellis_def}
\end{equation}
where $P$ is the path on $n$ nodes, and the minimization is over
mappings $\pi$ that are bijections between $I$ and $V(P)$.
From (\ref{treewidth_def}) and (\ref{trellis_def}), it is clear that 
$\k(\cC) \le \tau(\cC)$.

Let $v_1,v_2,\ldots,v_n$ be the nodes of the path $P$, listed in order from
one leaf to the other. For the bijection $\pi:I \to V(P)$ that maps
$i$ to $v_i$ ($1 \le i \le n$), we obtain from (\ref{eq:1}),
\begin{equation}
\k_{v_i} = k - \dim(\cC_{\pi[1,i-1]}) - \dim(\cC_{\pi[i+1,n]}),
\label{eq:6} 
\end{equation}
where $\pi[a,b] = \{\pi(j): a \le j \le b\}$.

\subsection{Generalized Hamming weights}
The generalized Hamming weights of a linear code, introduced and studied in 
\cite{Wei91}, limit the possible dimensions of shortened versions of the code. 
So, they are related to the complexity of tree realizations in a natural way.

Let $\cC$ be an $(n,k)$ linear code with index set $I$. We will use the 
notation $\cD \sqsubseteq \cC$ to say that $\cD$ is a subcode of $\cC$. 
For a subcode $\cD\sqsubseteq \cC$, we define its support 
$\chi(D)=\{i:\exists \, c_1c_2\ldots c_n \in \cD \text{ s.t.}\, c_i\ne0\}$. 
The $p$-th generalized Hamming weight of $C$, denoted $d_p(\cC)$, 
is the size of the smallest support of a $p$-dimensional subcode of $\cC$, 
i.e., $d_p(\cC)=\min\{|\chi(\cD)|:\cD\sqsubseteq \cC,\dim(\cD)=p\}$ 
for $1\le p\le k$. It is known that $0\le d_1(\cC)<d_2(\cC)<\cdots<d_k(\cC)\le n$. 
Also, $d_1(\cC)$ is the minimum distance of $\cC$.

A closely related definition is that of maximal limited-support subcode 
dimensions. For $1 \le s\le n$, $U_s(\cC)$ is defined to be the maximum 
dimension of a subcode of $\cC$ with support at most $s$, i.e., 
$U_s(\cC)=\max\{\dim(\cD):\cD\sqsubseteq \cC,|\chi(\cD)|\le s\}$. 
The maximal limited-support subcode dimensions can be computed 
using the generalized Hamming weights as follows:
\begin{equation}
U_{s}(\cC)=u \text{ such that } d_{u}(\cC) \le s <d_{u+1}(\cC)
\label{eq:U_s}
\end{equation}
with the convention that $d_{0}(\cC)=0$ and $d_{k+1}(\cC)=n+1$. 
We also define $U_0(\cC) = 0$.

\section{The Proof Strategy}
\label{sec:strategy}
From the relevant definitions, treewidth cannot exceed trelliswidth 
for any code $\cC$, i.e., $\k(\cC) \le \tau(\cC)$. 
We now describe a general strategy that can be used to show
the opposite inequality in certain cases.

Consider an $(n,k)$ linear code $\cC$, with index set $I$.
The idea of using maximal limited-support subcode dimensions 
to study the complexity of trellis realizations of $\cC$ was introduced in 
\cite{KTFL93}. We extend that idea to tree realizations here.
For $J \subseteq I$, $\cC_J$ is a subcode of $\cC$ with support at most $|J|$.
So, $\dim(\cC_J)\le U_{|J|}(\cC)$. Therefore, given any tree decomposition 
$(T,\om)$ of $\cC$, we obtain from (\ref{eq:1}) that for any $v \in V(T)$,
\begin{equation}
  \label{eq:10}
  \k_v \ge k-\sum_{e\in E(v)}U_{|I_{e,v}|}(\cC).
\end{equation}

Now, recall from the definition of treewidth that it suffices to carry out 
the minimization in (\ref{treewidth_def}) over tree decompositions 
$(T,\omega)$ in which $T$ is a cubic tree with $n$ leaves, and $\om$ is 
a bijection between $I$ and $L(T)$. For such a $(T,\omega)$,
we note that $|I_{e,v}|$ is simply the number of leaves in $T_{e,v}$,
and for an internal node $v \in V(T)$, the summation in (\ref{eq:10})
contains exactly three terms. 

Let $n_{e,v}$ denote the number of leaves in $T_{e,v}$, 
and note that these numbers $n_{e,v}$ are determined purely by 
the topology of $T$. At an internal node $v$ in a cubic tree $T$ 
with $n$ leaves, we will list the edges in $E(v)$ in the form of an ordered 
triple $[e_1(v)\;e_2(v)\;e_3(v)]$ such that 
$1<n_{e_1(v),v}\le n_{e_2(v),v}\le n_{e_3(v),v}<n$. 
If the node $v$ is clear in the context, we will use the simplified notation 
$n_i=n_{e_i(v),v}$ for $i=1,2,3$. 


Suppose that $T$ is a cubic tree with $n$ leaves having an internal node $v$
such that the numbers $n_1,n_2,n_3$ satisfy 
$\sum_{i=1}^3 U_{n_i}(\cC) \le k - \tau(\cC)$. Then, by (\ref{eq:10}),
for any bijection $\om$ between $I$ and $L(T)$,
we have $\k_v \ge \tau(\cC)$, and hence $\k(\cC;T,\om) \ge \tau(\cC)$.
Consequently, if every cubic tree with $n$ leaves had such a node $v$, 
then we would have $\k(\cC) \ge \tau(\cC)$. Since the opposite inequality
is always true, we have proved the following proposition.

\begin{prop}
Let $\cC$ be an $(n,k)$ linear code with the property that for any cubic tree 
$T$ with $n$ leaves, there always exists an internal node $v \in V(T)$ 
such that $\sum_{i=1}^3 U_{n_i}(\cC) \le k - \tau(\cC)$,
where $n_i = n_{e_i(v),v}$. Then, $\kappa(\cC)=\tau(\cC)$.
\label{prop:k_eq_t}
\end{prop}

A comment on the proof strategy implied by Proposition~\ref{prop:k_eq_t}
is in order. To show that $\k(\cC) \ge \tau(\cC)$ 
(and hence, $\k(\cC) = \tau(\cC)$), the obvious strategy
would be to show, for each tree decomposition $(T,\omega)$ of $\cC$,
the existence of a node $v \in V(T)$ for which $\k_v \ge \tau(\cC)$,
where $\k_v$ is given by (\ref{eq:10}). In general, the node $v$ would
depend on the tree $T$ as well as on the coordinate assignment $\omega$.
However, in the proof method based upon Proposition~\ref{prop:k_eq_t},
the idea is to find, for a given $(T,\omega)$, a node $v \in V(T)$ that 
depends only on the topology of $T$, and thus, is \emph{independent} 
of $\omega$, for which $\k_v \ge \tau(\cC)$ holds. It is 
a remarkable fact that this proof strategy can be made to work for
MDS and Reed-Muller codes, as we will see in Sections~\ref{sec:mds}
and \ref{sec:RM}.

The hypothesis of Proposition \ref{prop:k_eq_t} requires the 
existence of a node in any cubic tree, whose removal partitions the tree
into components with a certain property. The property in this case is
that the corresponding partition of the number of leaves, $n$, into
$n_1,n_2,n_3$ satisfies $\sum_{i=1}^3 U_{n_i}(\cC) \le k - \tau(\cC)$. 
Structural results of this form are known as separator theorems 
(see \emph{e.g.}, \cite{SW98})

A classical separator theorem is a theorem of Jordan \cite{Jor1869}
that states that any tree on $n$ nodes has an internal node whose 
removal leaves behind connected components with at most $n/2$ nodes each. 
A trivial modification of the simple proof of this theorem shows that 
the two occurrences of ``nodes'' in the theorem statement can be replaced 
by ``leaves''. For easy reference, we record this as a proposition for
the special case of cubic trees.

\begin{prop}
In any cubic tree with $n \ge 3$ leaves, there exists an internal node $v$ 
such that $n_{e_i(v),v} \le n/2$ for $i=1,2,3$.
\label{prop:n/2}
\end{prop}

Another classical (edge) separator theorem is the following result
(cf.\ \cite{SW98}):
every cubic tree $T$ with $n$ leaves contains an edge $e$ such that 
both components of $T-e$ have at most $2n/3$ leaves. 
Now, one of these two components must have at least $n/2$ leaves;
let $v$ be the node incident with $e$ for which this component is $T_{e,v}$.
Then, for this $v$, we have $n_3 \in [n/2,2n/3]$. We record this fact below.

\begin{prop} 
In any cubic tree with $n \ge 3$ leaves, there exists an internal node $v$ 
such that $n_{e_3(v),v} \in [n/2,2n/3]$.
\label{prop:cubic}
\end{prop}

As we will see in the next two sections, Propositions~\ref{prop:n/2} 
and \ref{prop:cubic} allow us to deal with MDS and Reed-Muller codes, 
respectively. We consider MDS codes first.

\section{Treewidth of MDS Codes}
\label{sec:mds}

MDS codes are $(n,k)$ linear codes for which the minimum distance equals 
$n-k+1$. Basic facts about MDS codes can be found in \cite{MS77}.

Let $\cC$ be an $(n,k)$ MDS code, with index set $I = [n]$.
The generalized Hamming weights of $\cC$ were computed in \cite{Wei91} 
as follows:
$$
  d_p(\cC)=n-k+p,\;\;1\le p\le k.
$$
From this, the maximal limited-support subcode dimensions, 
$U_{s}(\cC)$ for $1 \le s\le n$, can be determined using (\ref{eq:U_s}). 
They are given by
\begin{equation}
  \label{eq:mds_Us}
  U_s(\cC)=\begin{cases}
0,& 1 \le s \le n-k,\\
q,&s=n-k+q,\;q=1,2,\cdots,k.
\end{cases}
\end{equation}
Equivalently, $U_s(\cC)=\max\{0,s-(n-k)\}$. We use this to compute
$\tau(\cC)$ next.

Let $H$ be a parity-check matrix for $\cC$. 
For a subset $J \subseteq I$, the code $\cC_J$ has dimension 
equal to $|J| - \text{rank}(H|_J)$, where $H|_J$ refers to the 
restriction of $H$ to the columns indexed by $J$. As $\cC$ is MDS, 
$\text{rank}(H|_J) = \min\{|J|,n-k\}$. Hence, 
$\dim(\cC_J) = \max\{0,|J|-(n-k)\} = U_{|J|}(\cC)$. Therefore, 
for any permutation $\pi$ of $I$, we have for integers $1 \le a \le b \le n$,
$\dim(\cC_{\pi[a,b]}) = U_{b-a+1}(\cC)$. Therefore, 
the right-hand-side of (\ref{eq:6}) is always equal to 
$k - U_{i-1}(\cC) - U_{n-i}(\cC)$. It follows directly from this that
$$
\tau(\cC) = \max_{1 \le i \le n} (k - U_{i-1}(\cC) - U_{n-i}(\cC))
= k - \min_{1 \le i \le n} (U_{i-1}(\cC) + U_{n-1}(\cC)).
$$
A straightforward computation using (\ref{eq:mds_Us}) yields
$$
\min_{1\le i\le n}(U_{i-1}(\cC)+U_{n-i}(\cC))
=\begin{cases} 0,& \text{ if } n-k\ge k,\\
2k-n-1,& \text{ if } n-k<k.
\end{cases}
$$
achieved for $i=n-k+1$. We thus have the following result.

\begin{prop}
The trelliswidth of an $(n,k)$ MDS code $\cC$ is given by 
$\tau(\cC) = \min\{k,n-k+1\}$.
\label{prop:mds_trelliswidth}
\end{prop}

With this, we have
\begin{equation}
k - \tau(\cC) = \max\{0,2k-n-1\}.
\label{eq:mds1}
\end{equation}
%
%
We can now prove that the treewidth of an MDS code equals its trelliswidth.

\begin{theorem}
For an $(n,k)$ MDS code $\cC$, we have
$$
\k(\cC) = \tau(\cC) = \min\{k,n-k+1\}.
$$
\label{thm:mds_treewidth}
\end{theorem}
\begin{proof}
The statement is trivial for $n=1,2$, or when $k=n$, 
so we assume $n \ge 3$ and $1 \le n-k$. 
Let $T$ be a cubic tree with $n$ leaves, and let $v$ be the node guaranteed 
by Proposition~\ref{prop:n/2}. We will show that $v$ satisfies the hypothesis 
of Proposition~\ref{prop:k_eq_t}.

Set $n_i = n_{e_i(v),v}$, $i = 1,2,3$, and recall that, by definition, 
$n_1 \le n_2 \le n_3$. By choice of $v$, we also have $n_i \le n/2$ 
for $i=1,2,3$. For convenience, we write $U_{n_i}$ for $U_{n_i}(\cC)$.

\medskip

\noindent \underline{Case 1: $n-k \ge k$}.

In this case, $n_i \le n/2 \le n-k$, so that $\sum_i U_{n_i} = 0$ by 
(\ref{eq:mds_Us}). Moreover, by (\ref{eq:mds1}), $k-\tau(\cC) = 0$.

\medskip

\noindent \underline{Case 2: $1\le n-k < k$}.

Now, we have $n_i \le n/2 < k$. 
We must show that $\sum_{i} U_{n_i} \le 2k-n-1$. If $n_3 \leq n-k$,
then $\sum_i U_{n_i} = 0$. So, we assume $n_3 = k-\delta$,
with $1 \leq \delta < 2k-n$. Then, 
$U_{n_3}=n_3-(n-k)=2k-n-\delta$ and $n_1+n_2=n-n_3=n-k+\delta$. So, we have
  \begin{eqnarray*}
    U_{n_1}+U_{n_2}+U_{n_3}&=&\max\{0,k-n+n_1\}+\max\{0,k-n+n_2\}+2k-n-\delta\\
                 &\le& \max\{0,k-n+n_1,k-n+n_2,2k-2n+n_1+n_2\}+2k-n-\delta\\
                 &=&\max\{2k-n-\delta,3k-2n+n_2-\delta,3k-2n\}\\
                 &\le&2k-n-1,
  \end{eqnarray*}
where the last inequality holds because $\delta\ge1$, $n_2\le n-k+\delta-1$ 
and $n-k\ge1$. 

\medskip

Thus, in both cases, we see that $\sum_i U_{n_i} \le k - \tau(\cC)$,
and so, by Proposition~\ref{prop:k_eq_t}, we have $\k(\cC) = \tau(\cC)$.
\end{proof}

\section{Reed-Muller codes\label{sec:RM}}

For a positive integer $m$ and a non-negative integer $r$ with $0 \le r\le m$, 
the $r$-th order binary Reed-Muller code of length $2^m$, denoted RM$(r,m)$, 
is defined as follows. 
Let $P^m_r$ denote the set of all Boolean polynomials in $m$ variables of 
degree less than or equal to $r$. For an integer $i$, $0\le i\le 2^m-1$, 
with binary expansion $i=\sum_{j=0}^{m-1}b_j(i)2^j$, $b_j(i)\in\{0,1\}$, 
we let $\b(i)=(b_0(i),b_1(i),\cdots,b_{m-1}(i))$. For $f\in P^m_r$, 
let $f(\b(i))=f(b_0(i),b_1(i),\cdots,b_{m-1}(i))$. The code RM$(r,m)$ 
is defined as
\begin{equation}
  \label{eq:5}
  \text{RM}(r,m)=\{[f(\b(0))\;f(\b(1))\;\cdots\; f(\b(2^m-1))]: f\in P^m_r\}. 
\end{equation}
The code $\text{RM}(r,m)$ has length $n=2^m$, dimension $k(r,m) 
= \sum_{j=0}^r\binom{m}{j}$, and minimum distance $2^{m-r}$ \cite{MS77}. 
In (\ref{eq:5}), the order of evaluation of the function $f$ is 
according to the index set $I=[0,2^m-1]$. This is called the standard bit 
order. 

We will denote the treewidth and trelliswidth of $\text{RM}(r,m)$
by $\k(r,m)$ and $\tau(r,m)$, respectively.

\subsection{Trelliswidth of $\text{RM}(r,m)$}

Let $\cC$ be the Reed-Muller code $\text{RM}(r,m)$ in the standard bit order,
so that $I = [0,2^m-1]$. In this section, we derive an exact expression
for the trelliswidth of $\cC$.

Let $P$ be the path on $n=2^m$ nodes, with $v_0,v_1,\ldots,v_{n-1}$ being 
the nodes of $P$, listed in order from one leaf to the other.
For any $\pi: I \to V(P)$, we obtain from (\ref{eq:6}), in a manner analogous 
to the derivation of (\ref{eq:10}),
$$
\k_{v_i} \geq k(r,m) - U_i(\cC) - U_{n-1-i}(\cC),
$$
for $i = 0,1,\ldots,n-1$. Thus,
\begin{equation}
\k(\cC;P,\pi) \geq k(r,m) - \min_{0 \le i \le n-1} (U_i(\cC) + U_{n-1-i}(\cC)).
\label{eq:RM1}
\end{equation}
Note that the right-hand-side is independent of $\pi$, so that 
by (\ref{trellis_def}), 
\begin{equation}
\tau (\cC) \geq k(r,m) - \min_{0 \le i \le n-1} (U_i(\cC) + U_{n-1-i}(\cC)).
\label{eq:RM2}
\end{equation}
It is shown in \cite{KTFL93} that for $\text{RM}(r,m)$ in the standard bit 
order, we have for $i = 0,1,\ldots,n-1$,
\begin{equation}
  \dim(\cC_{[0,i]}))= U_{i+1}(\cC) \ \ \text{ and } \ \ 
  \dim(\cC_{[i,n-1]}) = U_{n-i}(\cC). 
\label{eq:std_bit_order}
\end{equation}
It follows that when $\pi$ simply maps $i$ to $v_i$ for all $i\in I$,
then we have equality in (\ref{eq:RM1}), and hence, in (\ref{eq:RM2}).
To put this another way, the branch complexity of the minimal trellis 
representation of $\text{RM}(r,m)$ in the standard bit order attains 
the lower bound on, and thus equals, the trelliswidth of the code. 
Techniques from \cite{BN00} allow us to compute, 
with very little effort, the branch complexity of this 
trellis representation. We give the details of this computation in
Appendix~A. From this, we obtain the following result.

\begin{prop} The trelliswidth of the Reed-Muller code $\text{RM}(r,m)$
is given by
$$
\tau(r,m) = 
\begin{cases}
\sum_{j=0}^r \binom{m-2j-1}{r-j} & \text{ if } m \ge 2r+1, \\
1 + \sum_{j=0}^{m-r-1} \binom{m-2j-1}{r-j} & \text{ if } m < 2r+1. \\
\end{cases}
$$
\label{prop:RM_tau}
\end{prop}

Recall that the dimension of the code $\text{RM}(r,m)$ is given by
$k(r,m) = \sum_{j=0}^r \binom{m}{j}$. We will find it convenient
to define $k(r',m')$ to be $\sum_{j=0}^{r'} \binom{m'}{j}$ 
for all non-negative integers $r',m'$, including when $r' > m'$. 
with the usual conventions that $\binom{0}{0} = 1$ and
$\binom{m'}{j} = 0$ for $j > m'$. Thus, for $r' \ge m' \ge 0$, 
$k(r',m') = 2^{m'}$. Following these conventions, we give an
expression for the difference $k(r,m) - \tau(r,m)$.

\begin{prop}
For the Reed-Muller code $\text{RM}(r,m)$, we have
$$
k(r,m) - \tau(r,m) 
 = \sum_{i=0}^{\min\{2(r-1),m-1\}} k(r-1-\lceil i/2 \rceil,m-1-i).
$$
\label{prop:Srm}
\end{prop}
We present the algebraic manipulations required to prove this
proposition in Appendix~A. 

\medskip

It is instructive to explicitly write out some of the terms of the
summation in the last proposition. When $m \ge 2r$, we have
\begin{eqnarray}
k(r,m) - \tau(r,m) \ = \ k(r-1,m-1)+ k(r-2,m-2) + k(r-2,m-3) && \notag \\
 +\  k(r-3,m-4) + k(r-3,m-5) && \notag \\
 +\ \cdots + k(0,m-2r+2) + k(0,m-2r+1),
\label{Srm_eq1}
\end{eqnarray}
and when $m \leq 2r-1$, we have
\begin{eqnarray}
k(r,m) - \tau(r,m) \ = \ k(r-1,m-1)+ k(r-2,m-2) + k(r-2,m-3) && \notag \\
 +\  k(r-3,m-4) + k(r-3,m-5) && \notag \\
 +\ \cdots + k(r-1-\lceil{\textstyle\frac{m-2}{2}}\rceil,1) 
  + k(r-1-\lceil{\textstyle\frac{m-1}{2}}\rceil,0).
\label{Srm_eq2}
\end{eqnarray}

\subsection{Treewidth of $\text{RM}(r,m)$}

We state below our main result showing that the treewidth of a Reed-Muller code
equals its trelliswidth.

\begin{theorem}
The treewidth of the Reed-Muller code $\text{RM}(r,m)$
is given by 
$$
\k(r,m) = \tau(r,m) = 
\begin{cases}
\sum_{j=0}^r \binom{m-2j-1}{r-j} & \text{ if } m \ge 2r+1, \\
1 + \sum_{j=0}^{m-r-1} \binom{m-2j-1}{r-j} & \text{ if } m < 2r+1. \\
\end{cases}
$$
\label{thm:RM_treewidth}
\end{theorem}

The rest of this section is devoted to a proof of the above result, 
which follows the strategy outlined in Section~\ref{sec:strategy}. 
Some of the technical details of the proof are presented in
Appendices B and C.

Let $\text{RM}(r,m)$ be given. If $m \le 2$, or $r = m$, then $\text{RM}(r,m)$
is an MDS code, which has been dealt with in Section~\ref{sec:mds}.
Henceforth, we will assume $m \ge 3$ and $r \le m-1$. 

Let $T$ be a cubic tree with $n = 2^m$ leaves, $m \ge 3$, 
and let $W = \{v \in V(T): n_{e_3(v),v} \in [n/2,2n/3]\}$. 
By Proposition~\ref{prop:cubic}, $W$ is non-empty. Let $v^* \in W$ 
be a node that achieves $\max\{n_{e_3(v),v}: v \in W\}$. 
Write $n_i^* = n_{e_i(v^*),v^*}$, $i = 1,2,3$.

\begin{lemma}
We have $n/6 < n_2^* < n/3.$
\label{lemma:n2*}
\end{lemma}
\begin{proof}
If $n_2^* < n/6$, then from the fact that $n_1^* \le n_2^*$, we obtain
$n_1^*+n_2^* < n/3$, so that $n_3^* > 2n/3$, a contradiction. 
So, $n_2 \ge n/6$. However, $n/6$ is not an integer for $n=2^m$, 
and so, $n_2^* > n/6$.

If $n_2^* \geq n/3$, then $n_1^*+n_3^* \leq 2n/3$. Let $v$ be the neighbour
of $v^*$ incident with edge $e_2(v^*)$. Then, setting $n_3 = n_{e_3(v),v}$,
we see that $n_3 = n_1^*+n_3^*$; see Figure~\ref{fig:1}. 
But this means that $n_3^* < n_3 \leq 2n/3$, which contradicts our
choice of $v^*$.
\end{proof}

\begin{figure}[t]
\centering \input{edge_traversal.pstex_t}
\caption{For $v$, we have $n_3 = n_1^* + n_3^*$.}
\label{fig:1}
\end{figure}
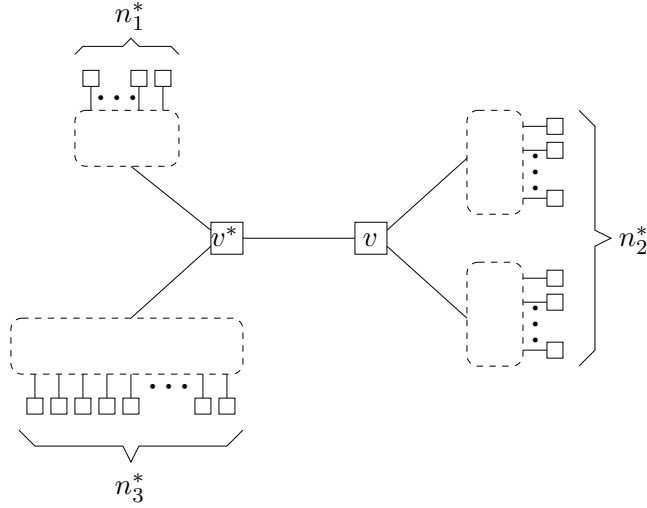

We will show that $\sum_{i=1}^3 U_{n_i^*} \le k(r,m) - \tau(r,m)$,
which will prove Theorem~\ref{thm:RM_treewidth}
by virtue of Proposition~\ref{prop:k_eq_t}.
Here, and in all that follows, we use $U_{h}$ as shorthand for 
$U_{h}(\text{RM}(r,m))$,

Denote by $\alpha^{(m)}$ and $\beta^{(m)}$ 
the largest integers in $[0,2n/3]$ and $[0,n/3]$, respectively.
Explicitly,
\begin{equation}
\alpha^{(m)} = \begin{cases}
\frac23 \cdot 2^m - \frac13 & \text{ if $m$ is odd }, \\
\frac23 \cdot 2^m- \frac23 & \text{ if $m$ is even},
\end{cases}
\label{eq:alpha}
\end{equation}
and
\begin{equation}
\beta^{(m)} = \begin{cases}
\frac13 \cdot 2^m - \frac23 & \text{ if $m$ is odd }, \\
\frac13 \cdot 2^m- \frac13 & \text{ if $m$ is even}.
\end{cases}
\label{eq:beta}
\end{equation}
Equivalently, in binary form, 
\begin{equation}
\b(\alpha^{(m)}) = \begin{cases}
(1,0,1,0,1,\cdots,0,1) & \text{ if $m$ is odd}, \\
(0,1,0,1,\cdots,0,1) & \text{ if $m$ is even}, \\
\end{cases}
\label{eq:binary_alpha}
\end{equation}
and
\begin{equation}
\b(\beta^{(m)}) = \begin{cases}
(0,1,0,1,0,\cdots,1,0) & \text{ if $m$ is odd}, \\
(1,0,1,0,\cdots,1,0) & \text{ if $m$ is even}. \\
\end{cases}
\label{eq:binary_beta}
\end{equation}
When there is no ambiguity, we will drop the superscripts from
$\alpha^{(m)}$ and $\beta^{(m)}$ for notational ease.

Now, what we know is that $n_3^* \in [2^{m-1},\alpha]$ and 
$n_2^* \in [\lceil\frac16 \, 2^m\rceil,\beta]$. In fact,
it can be directly verified from the expression for $\alpha$
that $\lceil\frac16 \, 2^m\rceil = \alpha-2^{m-1}+1$.
We wish to show that $\sum U_{n_i^*} \le k(r,m) - \tau(r,m)$. 
We will do this in two steps:
first, we show in Lemma~\ref{lemma:U_ineq} below that 
$\sum U_{n_i^*} \le U_\alpha+U_\beta+U_1$, and then, 
we prove in Lemma~\ref{lemma:Ua+Ub+U1} that 
$U_{\alpha} + U_{\beta} + U_1 = k(r,m) - \tau(r,m)$.

Write $n_3^* = \alpha - i$ and $n_2^* = \beta-j$, so that
$n_1^* = 2^m - (n_3^*+n_2^*) = i+j+1$, where 
$i \in [0,\alpha-2^{m-1}]$ and $j \in [0,\beta-(\alpha-2^{m-1}+1)]$.
The following lemma shows that $\sum U_{n_i^*} \le U_\alpha+U_\beta+U_1$.

\begin{lemma}
For $i \in [0,\alpha-2^{m-1}]$, and $j \in [0,\beta-(\alpha-2^{m-1}+1)]$,
we have 
$$
(U_{\alpha} - U_{\alpha-i}) + (U_\beta - U_{\beta-j}) 
\ge U_{i+j+1} - U_1.
$$
\label{lemma:U_ineq}
\end{lemma}
\begin{proof} 
See Appendix~B. 
\end{proof}

\begin{lemma}
$U_{\alpha} + U_{\beta} + U_1 = k(r,m) - \tau(r,m)$.
\label{lemma:Ua+Ub+U1}
\end{lemma}
\begin{proof}
The minimum distance of $\text{RM}(r,m)$ is $2^{m-r}$. Since we have assumed
$r \le m-1$, the minimum distance is at least 2, and hence, $U_1 = 0$.
In Appendix~C, we show the following: when $m \ge 2r$,
\begin{equation}
U_{s} = 
\begin{cases}
\sum_{i=0}^{r-1} k(r-1-i,m-1-2i) & \text{ if } s = \alpha, \\
\sum_{i=1}^{r-1} k(r-1-i,m-2i) & \text{ if } s = \beta. 
\end{cases}
\label{Us_eq1}
\end{equation}
Examining the above summations term-by-term, it may be verified that
the alternate terms on the right-hand side of (\ref{Srm_eq1}),
beginning with $k(r-1,m-1)$, sum to $U_\alpha$, 
while the remaining terms sum to $U_\beta$. Hence, when $m \ge 2r$,
the statement of the lemma holds.

When $m < 2r$, we show in Appendix~C that
\begin{equation}
U_{\alpha} = 
\begin{cases}
\sum_{i=0}^{\frac{m-1}{2}} k(r-1-i,m-1-2i) & \text{ if $m$ is odd}, \\
\sum_{i=0}^{\frac{m-2}{2}} k(r-1-i,m-1-2i) & \text{ if $m$ is even}. 
\end{cases}
\label{Us_eq2}
\end{equation}
and
\begin{equation}
U_{\beta} = 
\begin{cases}
\sum_{i=1}^{\frac{m-1}{2}} k(r-1-i,m-2i) & \text{ if $m$ is odd}, \\
\sum_{i=1}^{\frac{m}{2}} k(r-1-i,m-2i) & \text{ if $m$ is even}. 
\end{cases}
\label{Us_eq3}
\end{equation}
This time, it can be seen that the alternate terms on the right-hand side 
of (\ref{Srm_eq2}), beginning with $k(r-1,m-1)$, sum to $U_\alpha$, 
while the remaining terms sum to $U_\beta$. This completes the proof of 
the lemma.
\end{proof}

With this, the proof of Theorem~\ref{thm:RM_treewidth} is complete.

\section{Concluding Remarks}

In this paper, we proved the surprising fact that for the families
of MDS and Reed-Muller codes, if we use the maximum dimension of
local constraint codes to measure the complexity of a graphical realization,
then there is no advantage to be gained in going from trellis realizations
to cycle-free realizations on more complex tree topologies. This is 
particularly surprising for Reed-Muller codes, given that they 
have a natural binary-tree structure arising from the recursive
$|u|u+v|$ construction (see e.g.\ \cite{For03}). Of course, the situation
could be different if we used some other measure for the complexity of
a graphical realization, for example, the sum of the local constraint
dimensions.

It is also quite remarkable that the proof strategy outlined in
Section~\ref{sec:strategy} -- namely, identifying in any cubic tree $T$ 
a node $v \in V(T)$ such that $\k_v \ge \tau(\cC)$ for \emph{every} 
tree decomposition of the code $\cC$ on $T$ --- succeeds for MDS and 
Reed-Muller codes. As noted in that section, this strategy 
ignores the role played by the coordinate assignment $\omega$ 
in determining the local constraint code dimension, $\k_v$. It seems
unlikely that this method of proof would succeed for other code families.
It would of course be interesting to devise a set of tools that could be 
used to compute treewidth, or simply to determine whether or not treewidth 
can be strictly less than trelliswidth, for other families of algebraic codes.

\section*{Appendix A: Proofs of Propositions~\ref{prop:RM_tau} and 
\ref{prop:Srm}}

In this appendix, we compute the branch complexity of the minimal trellis
representation of $\text{RM}(r,m)$ in the standard bit order, from
which the expressions in Proposition~\ref{prop:RM_tau} and \ref{prop:Srm}
are obtained. We refer the reader to the survey by Vardy \cite{Var98} 
for the necessary background on the theory of trellis representations.

Let $\tau(r,m)$ and $\s(r,m)$ denote, respectively, the branch complexity 
and state complexity of the minimal trellis representation of $\text{RM}(r,m)$
in the standard bit order. Berger and Be'ery \cite{BB93} gave an explicit
expression for $\s(r,m)$:
$$
\s(r,m) = \sum_{j=0}^{\min\{r,m-r-1\}} \binom{m-2j-1}{r-j}.
$$
A different derivation of the above was given by Blackmore and Norton 
\cite{BN00}. We rely heavily on tools from \cite{BN00} to prove 
the following result, which is equivalent to Proposition~\ref{prop:RM_tau}.

\begin{prop}
$$
\tau(r,m) = 
\begin{cases} 
\s(r,m) & \text{ if } m \geq 2r+1, \\
\s(r,m) + 1 & \text{ if } m < 2r + 1.
\end{cases}
$$
\label{app_prop1}
\end{prop}

We introduce some terminology and notation that will be needed in
the proof of the proposition. Let $\cC$ be the 
code $\text{RM}(r,m)$ in the standard bit order, and let $n = 2^m$. 
Let $\cT$ be the minimal trellis of $\cC$. 
For $i=0,1,\ldots,n$, the dimension of the state space at depth $i$ in $\cT$
is denoted $\s_i$. Thus, $\s(r,m) = \max_i \s_i$. For $i = 0,1,\ldots,n-1$,
we denote by $\tau_i$ the dimension of the branch space between the 
state spaces at depths $i$ and $i+1$; then, $\tau(r,m) = \max_i \tau_i$.

The following definitions were made in \cite{BN00} for $0 \le i \le n-1$:
\begin{itemize}
\item[(a)] if $\dim(\cC_{[i+1,n-1]}) = \dim(\cC_{[i-1,n-1]}) - 1$, 
then $i$ is called a \emph{point of gain} of $\cC$; and
\item[(b)] if $\dim(\cC_{[0,i]}) = \dim(\cC_{[0,i-1]}) + 1$,
then $i$ is called a \emph{point of fall} of $\cC$.
\end{itemize}

As per our notation from Section~\ref{sec:RM}, $\b(i)$ denotes 
the $m$-bit binary representation of $i$, $0 \le i \le n-1$.
Let $|\b(i)|_0$ and $|\b(i)|_1$ denote the number of $0$s and $1$s,
respectively, in $\b(i)$. 

\begin{lemma}[\cite{BN00}, Proposition~2.2]
For $0 \le i \le n-1$,
\begin{itemize} 
\item[(a)] $i$ is a point of gain of $\cC$ iff $|\b(i)|_1 \le r$;
\item[(b)] $i$ is a point of fall of $\cC$ iff $|\b(i)|_0 \le r$.
\end{itemize}
\label{BN_lemma}
\end{lemma}

\medskip

\noindent \emph{Proof of Proposition~\ref{app_prop1}}.\ \
It is a fact that for any minimal trellis representation,
branch complexity either is equal to the state complexity or is exactly one
more than the state complexity. In particular, 
$\s(r,m) \le \tau(r,m) \le \s(r,m)+1$. So, to prove 
Proposition~\ref{app_prop1}, it suffices to show that 
\begin{equation}
\tau(r,m) = \s(r,m) + 1 \text{ iff } m \le 2r.
\label{app_eq1}
\end{equation}

Suppose that $\tau(r,m) = \tau_i$ for some $i \in [0,n-1]$. From the 
local behaviour of $\cT$ described in \cite[p.\ 44]{BN00}, it follows that
we can have $\tau_i = \s(r,m) + 1$ iff $\s_i = \s(r,m)$ and 
$i+1$ is a point of gain as well as a point of fall of $\cC$.

Thus, if $\tau_i = \s(r,m)+1$, then by Lemma~\ref{BN_lemma},
$m = |\b(i+1)|_1 + |\b(i+1)|_0 \le 2r$. This proves the ``only if''
direction of (\ref{app_eq1}).

Conversely, suppose $m \le 2r$. The proposition is clearly true if $m=r$, 
since $\text{RM}(m,m) = \{0,1\}^{2^m}$, and we have $\s(m,m) = 0$ and
$\tau(m,m) = 1$. So, we may assume $m \ge r+1$.
Take $i$ to be such that 
$\b(i) = (0,0\ldots,0,1,0,1,0,\ldots,1,0)$, with $|\b(i)|_1 = m-r-1$.
Then, by Theorem~2.11 in \cite{BN00}, $\s_i = \s(r,m)$.
Also, $\b(i+1) = (1,0,\ldots,0,1,0,1,0,\ldots,1,0)$, with
$|\b(i+1)|_1 = m-r \le r$ and $|\b(i+1)|_0 = m-(m-r) =r$. Hence, 
by Lemma~\ref{BN_lemma}, $i+1$ is a point of gain as well as a point of fall
of $\cC$. Hence, $\tau_i = \s(r,m) + 1$, which completes the proof of 
(\ref{app_eq1}), and hence, of Proposition~\ref{app_prop1}.
\qed

\medskip

We next present the algebraic manipulations needed to prove
Proposition~\ref{prop:Srm}.

\medskip

\noindent \emph{Proof of Proposition~\ref{prop:Srm}}.\ \ 
We divide the proof into three cases.

\noindent \underline{Case 1: $m \ge 2r+1$}.\ We have
\begin{eqnarray*}
k(r,m) - \tau(r,m) 
&=&  \sum_{j=0}^r \binom{m}{j} - \sum_{j=0}^r \binom{m-2j-1}{r-j} \\
&=& \sum_{j=0}^r \binom{m}{j} - \sum_{j=0}^r \binom{m-2(r-j)-1}{j} \\
&=& \sum_{j=1}^r \left[\binom{m}{j} - \binom{m-2(r-j)-1}{j}\right] \\
&\stackrel{(a)}{=}& \sum_{j=1}^r \sum_{i=0}^{2(r-j)} \binom{m-1-i}{j-1} \\
&\stackrel{(b)}{=}& \sum_{i=0}^{2(r-1)} \sum_{j=1}^{r-\lceil{i/2}\rceil} \binom{m-1-i}{j-1} \\
&=& \sum_{i=0}^{2(r-1)} k(r-1-\lceil{i/2}\rceil,m-1-i).
\end{eqnarray*}
In the above chain of equalities, equality~(a) uses the fact that
for integers $a < b$ and $j \ge 1$, we have $\binom{b}{j} - \binom{a}{j} 
= \sum_{q=a}^{b-1} \binom{q}{j-1}$; this is just repeated application
of the identity $\binom{b}{j} = \binom{b-1}{j-1} + \binom{b-1}{j}$.
Equality~(b) is obtained by exchanging the order of the summations in 
$i$ and $j$.

\smallskip

\noindent \underline{Case 2: $m = 2r$}.\ Here,
\begin{eqnarray*}
k(r,m) - \tau(r,m) 
&=&  \sum_{j=0}^r \binom{m}{j} - 1 - \sum_{j=0}^{r-1} \binom{m-2j-1}{r-j} \\
&=& \sum_{j=1}^r \binom{m}{j} - \sum_{j=1}^r \binom{m-2(r-j)-1}{j} \\
&=& \sum_{j=1}^r \left[\binom{m}{j} - \binom{m-2(r-j)-1}{j}\right],
\end{eqnarray*}
and now we carry on from equality~(a) of Case~1.

\smallskip

\noindent \underline{Case 3: $m \le 2r-1$}.\ This is the most tedious case.
We start with
\begin{eqnarray}
k(r,m) - \tau(r,m) 
&=&  \sum_{j=0}^r \binom{m}{j} - 1 - \sum_{j=0}^{m-r-1} \binom{m-2j-1}{r-j} 
\notag \\
&=& \sum_{j=1}^r \binom{m}{j} - \sum_{j=2r-m+1}^r \binom{m-2(r-j)-1}{j}
\notag \\
&=& \sum_{j=1}^{2r-m} \binom{m}{j} 
     + \sum_{j=2r-m+1}^r \left[\binom{m}{j} - \binom{m-2(r-j)-1}{j}\right] 
\notag \\
&=& \sum_{j=1}^{2r-m} \binom{m}{j} 
     + \sum_{j=2r-m+1}^r \sum_{i=0}^{2(r-j)} \binom{m-1-i}{j-1} \notag \\
&=&  \sum_{j=1}^{2r-m} \binom{m}{j} 
     + \sum_{i=0}^{2(m-r-1)} \sum_{j=2r-m+1}^{r-\lceil{i/2}\rceil} \binom{m-1-i}{j-1}.
\label{case3_eq1}
\end{eqnarray}
Now, for $j \ge 1$, write $\binom{m}{j} = \binom{m}{j} - \binom{0}{j}
= {\sum_{i=0}^{m-1} \binom{m-1-i}{j-1}}$. Hence, 
\begin{equation}
\sum_{j=1}^{2r-m} \binom{m}{j} 
 = \sum_{i=0}^{m-1} \sum_{j=1}^{2r-m} \binom{m-1-i}{j-1}.
\label{case3_eq2}
\end{equation}
Also,
\begin{equation}
\sum_{i=0}^{2(m-r-1)} \sum_{j=2r-m+1}^{r-\lceil{i/2}\rceil} \binom{m-1-i}{j-1}
= \sum_{i=0}^{m-1} \sum_{j=2r-m+1}^{r-\lceil{i/2}\rceil} \binom{m-1-i}{j-1},
\label{case3_eq3}
\end{equation}
as when $i \ge 2(m-r-1)+1$, we have $r-\lceil{i/2}\rceil \le 2r-m$,
so that the inner summation $\sum_{j=2r-m+1}^{r-\lceil{i/2}\rceil}$ is empty.
Plugging (\ref{case3_eq2}) and (\ref{case3_eq3}) into (\ref{case3_eq1}),
we find that 
$$
k(r,m) - \tau(r,m) 
= \sum_{i=0}^{m-1} \sum_{j=1}^{r-\lceil{i/2}\rceil} \binom{m-1-i}{j-1}
= \sum_{i=0}^{m-1} k(r-1-\lceil{i/2}\rceil,m-1-i).
$$
This completes the proof of Proposition~\ref{prop:Srm}. \qed

\section*{Appendix B: Proof of Lemma~\ref{lemma:U_ineq}}

We recast the statement of Lemma~\ref{lemma:U_ineq} into an equivalent
statement about binary representations of integers. 
From (\ref{eq:std_bit_order}) and the notion of points of fall
from \cite{BN00} (see Appendix~A), we see that for $1 \le s \le 2^m$, 
$U_s$ is equal to the number of points of fall of $\text{RM}(r,m)$ 
within the interval $[0,s-1]$. Thus, by Lemma~\ref{BN_lemma},
$U_s$ is equal to the number of integers in $[0,s-1]$ whose $m$-bit 
binary representations have at least $m-r$ $1$s. 

For an integer $j \in [0,2^m-1]$, let $\text{wt}(j)$ denote the Hamming weight
of (i.e., the number of 1s in) the binary representation $\b(j)$. 
For a subset $S \subseteq [0,2^m-1]$, let $w_i(S)$ denote the
number of integers $j \in S$ with $\text{wt}(j) \ge i$. 
We set $w_i(\emptyset) = 0$. Then, 
Lemma~\ref{lemma:U_ineq} is equivalent to the following assertion:
for $i \in [0,\alpha - 2^{m-1}]$ and 
 $j \in [0,\beta-(\alpha-2^{m-1} + 1)]$, we have
\begin{equation}
w_{m-r}([\alpha-i, \alpha-1]) 
 + w_{m-r}([\beta-j,\beta-1]) 
 \ge w_{m-r}([1,i+j]).
\label{appB_eq1}
\end{equation}
Since Lemma~\ref{lemma:U_ineq} needs to be shown for any $\text{RM}(r,m)$
with $0 \le r \le m-1$, we see that (\ref{appB_eq1}) must be shown
for any $m-r \in \{1,2,\ldots,m\}$. With this in mind, we define
for $S \subseteq [0,2^m-1]$,
$$
\w^{(m)}(S) = [w_1(S)\ w_2(S)\ \cdots\ w_m(S)].
$$
As usual, we will drop the superscript $(m)$ when it can be gleaned 
unambiguously from the context. 

\begin{prop}
For $m \ge 2$ and $0 \le i,j \le \alpha^{(m)} - 2^{m-1}$,
we have
\begin{equation}
\w^{(m)}([\alpha^{(m)}-i, \alpha^{(m)}-1]) 
 + \w^{(m)}([\beta^{(m)}-j,\beta^{(m)}-1]) 
 \ge \w^{(m)}([1,i+j]),
\label{eq:w_ineq}
\end{equation}
with the inequality above holding componentwise.
\label{appB_prop}
\end{prop}

Observe that this proposition is slightly stronger than 
Lemma~\ref{lemma:U_ineq}, since the latter only requires 
$0 \le j \le \beta^{(m)}-(\alpha^{(m)}-2^{m-1} + 1)$.
It is easy to verify that $\beta^{(m)}-(\alpha^{(m)}-2^{m-1} + 1) 
\le \alpha^{(m)} - 2^{m-1}$. The remainder of this appendix
is devoted to a proof of Proposition~\ref{appB_prop}. The proof 
is by induction on $m$, which is why we have taken care to
include the superscripts on $\alpha$ and $\beta$ in the statement
of the proposition. The main ingredients in the inductive proof
are the simple facts that for a non-negative integer $j$, 
$\text{wt}(2j) = \text{wt}(j)$ and $\text{wt}(2j+1) = \text{wt}(j) + 1$.
The rest is merely careful bookkeeping.

Let $P^{(m)}(i,j)$ denote the inequality in (\ref{eq:w_ineq}).
The induction argument is built upon certain implications among 
the $P^{(m)}(i,j)$, as stated in the series of lemmas below.
We introduce here some notation that we will use in the proofs
of these lemmas. 
For a set of integers $S$, we write $2S$ and $2S+1$ to mean 
the sets $\{2j: j \in S\}$ and $\{2j+1:j \in S\}$, respectively.
By $\1^{(m)}_{[a,b]}$, with $1 \le a \le b \le m$, we mean 
the vector $[z_1\ z_2\ \ldots\ z_m]$, with $z_i = 1$ for $a \le i \le b$,
and $z_i = 0$ otherwise. Again, we will drop the superscript $(m)$ 
when there is no ambiguity. 

\begin{lemma}
For even $m$, $P^{(m)}(i,j)$ implies $P^{(m+1)}(2i+1,2j)$. 
For odd $m$, $P^{(m)}(i,j)$ implies $P^{(m+1)}(2i,2j+1)$. 
\label{lem:1}
\end{lemma}
\begin{proof}
For even $m$, we have $\alpha^{(m+1)}=2\alpha^{(m)}+1$,
and $\beta^{(m+1)}=2\beta^{(m)}$. Set $S = [\alpha^{(m)}-i,\alpha^{(m)}-1]$
and $T = [\beta^{(m)}-j,\beta^{(m)}-1]$. Now, $P^{(m)}(i,j)$ implies
\begin{eqnarray}
\w^{(m+1)}(2S) + \w^{(m+1)}(2T) &\ge& \w^{(m+1)}(2[1,i+j]) 
 \label{eq:lem1.1}  \\
  \w^{(m+1)}(2S+1)+\w^{(m+1)}(2T+1)
  &\ge& \w^{(m+1)}(2[1,i+j]+1)   \label{eq:lem1.2}
\end{eqnarray}
since $\text{wt}(2j) = \text{wt}(j)$ and $\text{wt}(2j+1) = \text{wt}(j) + 1$ 
for any non-negative integer $j$. Henceforth, all the $\w$'s in this proof are
$\w^{(m+1)}$'s. Combining (\ref{eq:lem1.1}) and (\ref{eq:lem1.2}), we have
$$
\mathbf{w}([2\alpha^{(m)}-2i,2\alpha^{(m)}-1])
 +\mathbf{w}([2\beta^{(m)}-2j,2\beta^{(m)}-1])
\ge
\mathbf{w}([2,2i+2j+1]),$$
which is the same as
\begin{equation}
  \label{eq:lem1.3}
 \mathbf{w}([\alpha^{(m+1)}-2i-1,\alpha^{(m+1)}-2])
 +\mathbf{w}([\beta^{(m+1)}-2j,\beta^{(m+1)}-1])\ge\mathbf{w}([2,2i+2j+1]).
\end{equation}
Now, $\mathbf{w}([1,2i+2j+1])=\mathbf{w}([2,2i+2j+1])+\1^{(m+1)}_{[1,1]}$. 
Also, $\mathbf{w}([\alpha^{(m+1)}-2i-1,\alpha^{(m+1)}-1])
=\mathbf{w}([\alpha^{(m+1)}-2i-1,\alpha^{(m+1)}-2])+\1^{(m+1)}_{[1,m/2]}$, 
since $\text{wt}(\alpha^{(m+1)}-1) 
= \text{wt}(2\alpha^{(m)}) = \text{wt}(\alpha^{(m)}) = m/2$, 
by (\ref{eq:binary_alpha}). Therefore,
\begin{equation}
  \label{eq:lem1.4}
  \w([a^{(m+1)}-2i-1,a^{(m+1)}-1])+\w([b^{(m+1)}-2j,b^{(m+1)}-1])
   \ge \w([1,2i+2j+1]),
\end{equation}
which is $P^{(m+1)}(2i+1,2j)$. 

The proof for odd $m$ is along similar lines.
\end{proof}

\medskip

\begin{lemma}
\begin{itemize}
\item[(a)] When $\alpha^{(m)}-i$ is even, the two inequalities 
$P^{(m)}(i,j)$ and $P^{(m)}(i+2,j)$ together imply $P^{(m)}(i+1,j)$.
\item[(b)] When $\beta^{(m)}-j$ is even, the two inequalities 
$P^{(m)}(i,j)$ and $P^{(m)}(i,j+2)$ together imply $P^{(m)}(i,j+1)$.
\end{itemize}
\label{lem:2}
\end{lemma}

\begin{proof}
We only prove (a), as the proof of (b) is completely analogous.
In this proof, all omitted superscripts are to be taken to be $(m)$. 

Let $x=\text{wt}(\alpha-i-1)$ and $y=\text{wt}(i+j+1)$. We have
$\w([\alpha-i-1,\alpha-1]) = \w([\alpha-i,\alpha-1]) + \1_{[1,x]}$,
and $\w([1,i+j+1]) = \w([1,i+j]) + \1_{[1,y]}$.
We want to show $P^{(m)}(i+1,j)$:
\begin{equation}
\label{eq:lem2.1}
 \w([\alpha-i,\alpha-1]) + \1_{[1,x]} + \w([\beta-j,\beta-1])
 \ge \w([1,i+j]) + \1_{[1,y]}.
\end{equation}

If $x\ge y$, then $P^{(m)}(i,j)$ clearly implies (\ref{eq:lem2.1}).
So, suppose $x<y$. Then, (\ref{eq:lem2.1}) becomes
$$
 \w([\alpha-i,\alpha-1]) + \w([\beta-j,\beta-1])
 \ge \w([1,i+j]) + \1_{[x+1,y]},
$$
or equivalently,
\begin{equation}
\label{eq:lem2.2}
 w_l([\alpha-i,\alpha-1]) + w_l([\beta-j,\beta-1])
 \ge 
\begin{cases}
w_l([1,i+j]) + 1 & \text{ if } x+1 \le l \le y \\
w_l([1,i+j]) & \text{ otherwise}.
\end{cases}
\end{equation}

Let $x'=\text{wt}(\alpha-i-2)$ and $y'=\text{wt}(i+j+2)$. 
Since $\alpha-i$ is even, we see that $x'+1=x$ or $x'<x$. Now, we have
\begin{eqnarray}
\label{eq:lem2.3}
 \w([\alpha-i-2,\alpha-1])&=&\w([\alpha-i,\alpha-1])+\1_{[1,x]}+\1_{[1,x']} \\
\label{eq:lem2.4} 
 \w([1,i+j+2])&=&\w([1,i+j])+\1_{[1,y]}+\1_{[1,y']}
\end{eqnarray}
Thus, $P^{(m)}(i+2,j)$ is equivalent to
\begin{equation}
\label{eq:lem2.5}
 \w([\alpha-i,\alpha-1])+\1_{[1,x']} + \w([\beta-j,\beta-1])
 \ge \w([1,i+j]) + \1_{[x+1,y]}+\1_{[1,y']}.
\end{equation}
Using the fact that $x'<x$, (\ref{eq:lem2.5}) implies that 
for $x+1 \le l \le y$, 
$$
 w_l([\alpha-i,\alpha-1]) + w_l([\beta-j,\beta-1]) \ge w_l([1,i+j]) + 1.
$$
Since $P^{(m)}(i,j)$ clearly implies the ``otherwise'' part of 
(\ref{eq:lem2.2}), we have shown that $P^{(m)}(i,j)$ and $P^{(m)}(i+2,j)$
together imply (\ref{eq:lem2.2}), i.e., $P^{(m)}(i+1,j)$.
\end{proof}

\medskip

\begin{lemma}
For even $m$, the following implications hold:
\begin{itemize}
\item[(a)] $P^{(m)}(i,j) \ \Longrightarrow \ P^{(m+1)}(2i+1,2j)$;
\item[(b)] $P^{(m)}(i-1,j) \wedge P^{(m)}(i,j) \
\Longrightarrow \ P^{(m+1)}(2i,2j)$;
\item[(c)] $P^{(m)}(i,j) \wedge P^{(m)}(i,j+1) \
\Longrightarrow \ P^{(m+1)}(2i+1,2j+1)$;
\item[(d)] $P^{(m)}(i-1,j) \wedge P^{(m)}(i,j) 
\wedge P^{(m)}(i-1,j+1) \wedge P^{(m)}(i,j+1) 
\ \Longrightarrow \ P^{(m+1)}(2i,2j+1)$.
\end{itemize}
\label{lem:3}
\end{lemma}
\begin{proof}
(a) follows directly from Lemma~\ref{lem:1}. 

(b): If $P^{(m)}(i-1,j)$ and $P^{(m)}(i,j)$ are true, 
then by Lemma~\ref{lem:1}, we have $P^{(m+1)}(2i-1,2j)$ 
and $P^{(m+1)}(2i+1,2j)$ being true. Since $m+1$ is odd,
$\alpha^{(m+1)}$ is odd (see (\ref{eq:binary_alpha})). 
It now follows from Lemma~\ref{lem:2}(a) that $P^{(m+1)}(2i,2j)$ holds.

(c): This follows by an argument similar to part~(b), 
except that Lemma~\ref{lem:2}(b) is applied.

(d): By part~(b), $P^{(m+1)}(2i,2j)$ and  $P^{(m+1)}(2i,2j+2)$ hold.
Therefore, by Lemma~\ref{lem:2}(b),  $P^{(m+1)}(2i,2j+1)$ holds.
\end{proof}

\medskip

Arguments similar to those used in the above proof show the next result.

\begin{lemma}
For odd $m$, the following implications hold:
\begin{itemize}
\item[(a)] $P^{(m)}(i,j) \ \Longrightarrow \ P^{(m+1)}(2i,2j+1)$;
\item[(b)] $P^{(m)}(i,j-1) \wedge P^{(m)}(i,j) \
\Longrightarrow \ P^{(m+1)}(2i,2j)$;
\item[(c)] $P^{(m)}(i,j) \wedge P^{(m)}(i+1,j) \
\Longrightarrow \ P^{(m+1)}(2i+1,2j+1)$;
\item[(d)] $P^{(m)}(i,j-1) \wedge P^{(m)}(i,j) 
\wedge P^{(m)}(i+1,j-1) \wedge P^{(m)}(i+1,j) 
\ \Longrightarrow \ P^{(m+1)}(2i+1,2j)$.
\end{itemize}
\label{lem:4}
\end{lemma}

We are now in a position to prove Proposition~\ref{appB_prop}.

\smallskip

\noindent\emph{Proof of Proposition~\ref{appB_prop}\/}.\ \
Set $\ell^{(m)} = \alpha^{(m)} - 2^{m-1}$. We wish to show that for
$m \ge 2$, $P^{(m)}(i,j)$ holds for $0 \le i,j \le \ell^{(m)}$. 
It is easy to verify this directly for $m=2$ and $m=3$, 
so we start the induction by assuming that for some odd $m \ge 3$, 
$P^{(m)}(i,j)$ holds for $0 \le i,j \le \ell^{(m)}$. 

For odd $m$, the implications in Lemma~\ref{lem:4} 
are enough to show that $P^{(m+1)}(i,j)$ holds for $1 \le i \le 2\ell^{(m)}$
and $1 \le j \le 2\ell^{(m)}+1$. Note also that for odd $m$, 
we have $\ell^{(m+1)} = 2\ell^{(m)}$, as can be verified from (\ref{eq:alpha}).
Since $P^{(m+1)}(0,0)$, $P^{(m+1)}(0,1)$ and $P^{(m+1)}(1,0)$ 
trivially hold, we have that $P^{(m+1)}(i,j)$ holds for 
$0 \le i \le \ell^{(m+1)}$ and $0 \le j \le \ell^{(m+1)}+1$.

Now, $m+1$ is even, and we have shown above that $P^{(m+1)}(i,j)$ is true 
for $0 \le i \le \ell^{(m+1)}$ and $0 \le j \le \ell^{(m+1)}+1$. 
The implications in Lemma~\ref{lem:3} are then sufficient to show that
$P^{(m+2)}(i,j)$ holds for $1 \le i,j \le 2\ell^{(m+1)} + 1$. Again,
$P^{(m+2)}(0,0)$, $P^{(m+2)}(0,1)$ and $P^{(m+2)}(1,0)$ can be seen to hold
trivially, so $P^{(m+2)}(i,j)$ in fact holds for 
$0 \le i,j \le 2\ell^{(m+1)} + 1$. This completes the induction step,
since for even $m+1$, it follows from (\ref{eq:alpha}) that
$\ell^{(m+2)} = 2\ell^{(m+1)}+1$. \qed

\medskip

As observed earlier, Proposition~\ref{appB_prop} proves 
Lemma~\ref{lemma:U_ineq}.

\section*{Appendix C: Computing $U_\alpha$ and $U_\beta$}

To derive the expressions in (\ref{Us_eq1})--(\ref{Us_eq3}),
we make use of (\ref{eq:U_s}) and a result of Wei \cite{Wei91} that 
explicitly determines the generalized Hamming weight hierarchy of 
$\text{RM}(r,m)$. Any non-negative integer $u < k(r,m)$ can be uniquely
expressed as a sum 
\begin{equation}
u = \sum_{i=1}^{\ell} k(r_i,m_i),
\label{eq:canonical}
\end{equation}
where $r > r_1 \ge r_2 \ge \ldots \ge r_\ell \ge 0$, 
$m > m_1 \ge m_2 \ge \ldots \ge m_\ell \ge 0$, 
and for all $i$, $m_i - r_i = m-r+1-i$ \cite[Lemma~2]{Wei91}.
The above representation is called the 
\emph{$(r,m)$-canonical representation} of $u$.

\begin{theorem}[\cite{Wei91}, Corollary~6]
For $0 \le u < k(r,m)$, given the unique $(r,m)$-canonical representation
of $u$ as in (\ref{eq:canonical}), we have 
$d_u(\text{RM}(r,m))=\sum_{i=1}^\ell 2^{m_i}$.
\label{thm:wei}
\end{theorem}
For convenience, we will henceforth write $d_u(\text{RM}(r,m))$
simply as $d_u$.

\smallskip

Assume that $m \ge 2r$. We want to show that (\ref{Us_eq1}) holds. We 
will only prove here the result for $s=\alpha$, as the result for 
$s=\beta$ can be proved analogously. Let $\hu$ be the integer given by 
\begin{equation}
\hu = \sum_{i=1}^{r} k(r-i,m+1-2i).
\label{eq:hu}
\end{equation}
Note that the above is the $(r,m)$-canonical representation of $\hu$.
By Theorem~\ref{thm:wei}, we have $d_{\hu} = \sum_{i=1}^r 2^{m+1-2i}$. 
In binary form, $\b(d_{\hu}) = (0,0,\ldots,0,0,1,0,1,\ldots,0,1)$, 
the number of 1s in $\b(d_{\hu})$ being $r$. Comparing this with 
the binary form of $\alpha$ given in (\ref{eq:binary_alpha}), 
it is clear that $d_{\hu} \le \alpha$.

Next, write $\hu + 1$ as 
$$
\hu + 1 = \sum_{i=1}^{r} k(r-i,m+1-2i) + k(0,m-2r),
$$
using the fact that $k(0,m-2r) = \binom{m-2r}{0} = 1$. This is again
in $(r,m)$-canonical form, and hence by Theorem~\ref{thm:wei},
we have $d_{\hu + 1} = \sum_{i=1}^r 2^{m+1-2i} + 2^{m-2r}$. 
In binary form, this is $\b(d_{\hu + 1}) = (0,0,\ldots,0,1,$
$1,0,1,\ldots,0,1)$, the number of 1s here being $r+1$. 
Comparing with (\ref{eq:binary_alpha}), we see that $\alpha < d_{\hu+1}$.

Since $d_{\hu} \le \alpha < d_{\hu+1}$, we have by (\ref{eq:U_s}),
$U_\alpha = \hu$. Observe that $\hu$ as given by (\ref{eq:hu})
is precisely equal to the claimed value of $U_\alpha$ in (\ref{Us_eq1}). 

Now, assume $m < 2r$. We wish to show (\ref{Us_eq2}) and (\ref{Us_eq3}).
We sketch the proof for (\ref{Us_eq3}) here; the proof for 
(\ref{Us_eq2}) is similar. Set 
$$
\vu = 
\begin{cases}
\sum_{i=1}^{\frac{m-1}{2}} k(r-1-i,m-2i) & \text{ if $m$ is odd}, \\
\sum_{i=1}^{\frac{m}{2}} k(r-1-i,m-2i) & \text{ if $m$ is even}.
\end{cases}
$$
The above is the $(r,m)$-canonical representation of $\vu$,
and hence, 
$$
d_{\vu} = 
\begin{cases}
\sum_{i=1}^{\frac{m-1}{2}} 2^{m-2i} & \text{ if $m$ is odd}, \\
\sum_{i=1}^{\frac{m}{2}} 2^{m-2i} & \text{ if $m$ is even}.
\end{cases}
$$
Comparing $\b(d_{\vu})$ with $\b(\beta)$ given in (\ref{eq:binary_beta}),
it can be seen that $d_{\vu} \le \beta$.

The $(r,m)$-canonical representation of $\vu+1$ is given by
$$
\begin{cases}
\sum_{i=1}^{\frac{m-1}{2}} k(r-1-i,m-2i) + k(r-1-\frac{m-1}{2},0)
& \text{ if $m$ is odd}, \\
\sum_{i=1}^{\frac{m}{2}-1} k(r-1-i,m-2i) + k(r-\frac{m}{2},1) & 
\text{ if $m$ is even}. 
\end{cases}
$$
Again, $d_{\vu+1}$ can be obtained from Theorem~\ref{thm:wei}, and
the subsequent comparison of binary forms shows that $\beta < d_{\vu+1}$.
Hence, by (\ref{eq:U_s}), we have $U_\beta = \vu$, which proves 
(\ref{Us_eq3}).

\section*{Acknowledgement} 
A.\ Thangaraj thanks Rakesh Pokala for several helpful discussions.

\end{document}

%% file: edge_traversal.pstex_t
\begin{picture}(0,0)%
\epsfig{file=edge_traversal.pstex}%
\end{picture}%
\setlength{\unitlength}{2644sp}%
\begingroup\makeatletter\ifx\SetFigFont\undefined%
\gdef\SetFigFont#1#2#3#4#5{%
  \reset@font\fontsize{#1}{#2pt}%
  \fontfamily{#3}\fontseries{#4}\fontshape{#5}%
  \selectfont}%
\fi\endgroup%
\begin{picture}(6573,4656)(1114,-3724)
\put(3001,-1336){\makebox(0,0)[lb]{\smash{{\SetFigFont{11}{13.2}{\familydefault}{\mddefault}{\updefault}{\color[rgb]{0,0,0}$v^*$}%
}}}}
\put(4426,-1336){\makebox(0,0)[lb]{\smash{{\SetFigFont{11}{13.2}{\familydefault}{\mddefault}{\updefault}{\color[rgb]{0,0,0}$v$}%
}}}}
\put(2101,764){\makebox(0,0)[lb]{\smash{{\SetFigFont{11}{13.2}{\familydefault}{\mddefault}{\updefault}{\color[rgb]{0,0,0}$n_1^*$}%
}}}}
\put(6826,-1336){\makebox(0,0)[lb]{\smash{{\SetFigFont{11}{13.2}{\familydefault}{\mddefault}{\updefault}{\color[rgb]{0,0,0}$n_2^*$}%
}}}}
\put(2101,-3661){\makebox(0,0)[lb]{\smash{{\SetFigFont{11}{13.2}{\familydefault}{\mddefault}{\updefault}{\color[rgb]{0,0,0}$n_3^*$}%
}}}}
\end{picture}%